\documentclass{article}
\author{\textbf{David Bremner and Rasoul Shahsavarifar} \vspace{.5cm}
\\Faculty of Computer Science\\University of New Brunswick\\ Fredericton, NB, Canada\\Email:\{bremner, ra.shahsavari\} @unb.ca}
\title{\textbf{An Optimal Algorithm for Computing the Spherical Depth of Points in the Plane}}
\date{\today}

\usepackage{hyperref}
\usepackage{fullpage}
\usepackage{amsmath}
\usepackage{enumerate}
\usepackage{url}
\usepackage{algorithm}
\usepackage{algorithmicx}
\usepackage[noend]{algpseudocode}
\usepackage{pifont}
\makeatletter
\def\BState{\State\hskip-\ALG@thistlm}
\makeatother
\usepackage{polynom}
\usepackage{subfig}
\usepackage{amsfonts}
\usepackage[all]{xy}
\usepackage{graphicx}
\usepackage{titlesec}
\usepackage{color}
\usepackage{tabularx}
\usepackage{caption}
\usepackage{array}
\usepackage{algorithm}
\usepackage{pifont}
\makeatletter
\def\BState{\State\hskip-\ALG@thistlm}
\makeatother
\usepackage[english]{babel}
\usepackage{amsthm}
\usepackage[english]{babel}
\theoremstyle{definition}

\newtheorem{theorem}{Theorem}[section]

\newtheorem{lemma}[theorem]{Lemma}

\usepackage{amsopn}

\DeclareMathOperator{\SD}{\textit{SD}}

\DeclareMathOperator{\SphD}{\textit{SphD}}
\DeclareMathOperator{\Sph}{\textit{Sph}}
\DeclareMathOperator{\Bin}{\textit{Bin}}
\DeclareMathOperator{\Sin}{\textit{Sin}}

\begin{document}
\maketitle
\vspace{1.5cm}
\paragraph{Abstract.} \label{sec:abst}
For a distribution function $F$ on $\mathbb{R}^d$ and a point $q\in \mathbb{R}^d$, the \emph{spherical depth} $\SphD(q;F)$ is defined to be the probability that a point $q$ is contained inside a random closed hyper-ball obtained from a pair of points from $F$. The spherical depth $\SphD(q;S)$ is also defined for an arbitrary data set $S\subseteq \mathbb{R}^d$ and $q\in \mathbb{R}^d$. This definition is based on counting all of the closed hyper-balls, obtained from pairs of points in $S$, that contain $q$. The significant advantage of using the spherical depth in multivariate data analysis is related to its complexity of computation. Unlike most other data depths, the time complexity of the spherical depth grows linearly rather than exponentially in the dimension $d$. The straightforward algorithm for computing the spherical depth in dimension $d$ takes $O(dn^2)$. The main result of this paper is an optimal algorithm that we present for computing the bivariate spherical depth. The algorithm takes $O(n \log n)$ time. By reducing the problem of \textit{Element Uniqueness}, we prove that computing the spherical depth requires $\Omega(n \log n)$ time. Some geometric properties of spherical depth are also investigated in this paper. These properties indicate that \emph{simplicial depth} ($\SD$) (Liu, 1990) is linearly bounded by spherical depth (in particular, $\SphD\geq \frac{2}{3}SD$). To illustrate this relationship between the spherical depth and the simplicial depth, some experimental results are provided. The obtained experimental bound ($\SphD\geq 2\SD$) indicates that, perhaps, a stronger theoretical bound can be achieved.

\section{Introduction} \label{sec:intro}
The rank statistic tests play an important role in univariate non-parametric statistics. If one attempts to generalize the rank tests to the multivariate case, the problem of defining a multivariate order statistic will occur. It is not clear how to define a multivariate order or rank statistic in a meaningful way. One approach to overcome this problem is to use the notion of data depth. Data depth measures the centrality of a point in a given data set in non-parametric multivariate data analysis. In other words, it indicates how deep a point is located with respect to the data set.
\\\\Over the last decades, various notions of data depth such as \emph{halfspace depth} (Hotelling, 1929,  \cite{hotelling1990stability,small1990survey}; Tukey, 1975, \cite{tukey1975mathematics}), \emph{simplicial depth} (Liu, 1990, \cite{liu1990notion}) \emph{Oja depth} (Oja, 1983, \cite{oja1983descriptive}), and others have emerged as powerful tools for non-parametric multivariate data analysis. Most of them have been defined to solve specific problems in data analysis. They are different in application, definition, and geometry of their central regions (regions with the maximum depth).
\\\\In 2006, Elmore, Hettmansperger, and Xuan \cite{elmore2006spherical} defined another notion of data depth named \emph{spherical depth}. It is defined as the probability that point $q$ is contained in a closed random hyper-ball with the diameter $\overline{x_ix_j}$, where $x_i$ and $x_j$ are two random points from a common distribution function $F$. These closed hyper-balls are known as influence regions of the spherical depth function. The concept of sphere area is the multidimensional generalization of \emph{Gabriel circles} in the definition of the \emph{Gabriel Graph} \cite{liu2011lens}. Spherical depth has some nice properties including affine invariance, symmetry, maximality at the centre and monotonicity. All of these properties are explored in \cite{elmore2006spherical} and \cite{yang2014depth}.
\\\\A notable characteristic of the spherical depth is that its time complexity grows linearly in dimension $d$ while for most other data depths the time complexity grows exponentially. To the best of our knowledge, the current best algorithm for computing the spherical depth is the straightforward algorithm which takes $O(n^2)$.
\\\\In this paper, we present an $O(n\log n)$ algorithm for computing the spherical depth in $\mathbb{R}^2$. Furthermore, we reduce the problem of Element Uniqueness\footnote{Element Uniqueness problem: Given a set $A=\{a_1, a_2,...,a_n\}$, is there a pair of indices $i,j$ with $i \neq j$ such that $a_i = a_j$?} to prove that computing the spherical depth of a query point requires $\Omega (n\log n)$ time. We also  investigate some geometric properties of spherical depth. These properties lead us to bound the simplicial depth of a point in terms of the spherical depth. Finally, some experiments are provided to illustrate the relationship between spherical depth and simplicial depth.

\section{Spherical depth}

\paragraph{Definition:} The spherical influence region, also called sphere area, of $x_i$ and $x_j$ in $\mathbb{R}^d$ ($Sph(x_i, x_j)$) is a closed hyperball with the diameter $\overline{x_ix_j}$. In other words,
\[ \forall (i,j) : Sph(x_i,x_j)=\left\{t\mid d(t,\frac{x_i+x_j}{2}) \leq \frac{d(x_i , x_j)}{2}\right\},
\]
where $d(.,.)$ is the Euclidean distance. Figure \ref{fig:sphere} shows the sphere area formed by two points $x_i$ and $x_j$, $Sph(x_i, x_j)$, in $\mathbb{R}^2$.

\paragraph{Definition:}
For a distribution function $F$ on $\mathbb{R}^d$, the spherical depth function of $q \in \mathbb{R}^d$ is defined as the probability that $q$ is contained within the sphere area $\Sph(x_i, x_j)$ of two random vectors $x_i$ and $x_j$ from $F$. This definition can be represented by (\ref{eq:sph-distribution}).

\begin{equation}
\label{eq:sph-distribution}
\SphD(q;F)= P(q \in Sph(x_i, x_j))
\end{equation}

\paragraph{Definition:}
Let $S=\{x_1,...,x_n\}$ be a set of points in $\mathbb{R}^d$. $\SphD(q;S)$, the spherical depth of a point $q \in \mathbb{R}^d$ with respect to $S$, is defined as a proportion of the sphere areas of $\Sph(x_i, x_j), 1\leq i < j \leq n$ that contain $q$. Using the indicator function $\textbf{\textit{I}}$, this definition can be represented by (\ref{eq:spherical}).

\begin{equation}
\label{eq:spherical}
\SphD(q;S)= \frac{1}{{n \choose 2}}\sum_{1\leq i<j\leq n}^{n} {I(q \in \Sph(x_i, x_j)})
\end{equation}

Figure \ref{fig:3points1-sph} shows the spherical depth of an arbitrary point $q \in \mathbb{R}^2$  with respect to $S=\{p_1, p_2, p_3 \}$.

\begin{figure}[!ht]
\centering
\parbox{5cm}{
\includegraphics[width=5cm]{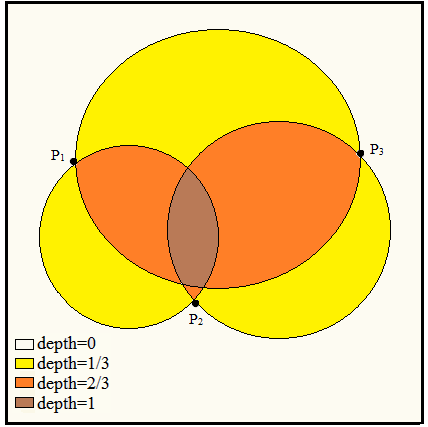}
\caption{Spherical depth of the points in the plane}
\label{fig:3points1-sph}}
\qquad
\begin{minipage}{5cm}
\includegraphics[width=6.5cm]{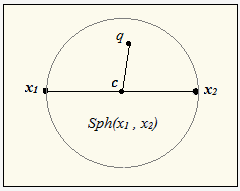}
\caption{\emph{Sphere area defined by $x_i$ and $x_j$}}
\label{fig:sphere}
\end{minipage}
\end{figure}

\subsection{Algorithm for Computing the Spherical Depth of a Query Point}
The current best algorithm for computing the spherical depth of a point $q\in \mathbb{R}^d$ with respect to a data set $S=\{x_1, x_2,...,x_n\}\subseteq \mathbb{R}^d$ is the brute force algorithm. This naive algorithm needs to check all of the $n \choose 2$ sphere areas obtained from the data points to figure out how many of them contain $q$. It can be verified that $q$ is contained in the sphere area $\Sph(x_i, x_j)$ if and only if $d(x_i,x_j) \geq 2 d(q,c)$, where $x_i$ and $x_j$ are two arbitrary points from $S$, and $c=(x_i+x_j)/2$ is the center of $\Sph(x_i, x_j)$. See Figure \ref{fig:sphere} for an illustration in $\mathbb{R}^2$. Checking all of the sphere areas causes the naive algorithm to take $\Theta(dn^2)$. Instead of counting, we focus on the geometric aspects of the sphere areas. These geometric properties lead us to develop an $O(n\log n)$ algorithm (Algorithm \ref{Alg:sph- pseudocode}) for the computation of the spherical depth of $q$.

\begin{lemma}
For arbitrary points $a$, $b$, and $t$ in $\mathbb{R}^2$, $t \in \Sph(a,b)$ if and only if $\angle atb \geq \frac{\pi}{2}$.
\label{lm:point-circle}
\end{lemma}

\begin{proof}
If $t$ is on the boundary of $\Sph(a,b)$, the \emph{Inscribed Angle Theorem} (Theorem $2.2$ in \cite{libeskind2008euclidean}) suffices as the proof in both directions. For the rest of the proof, by $t \in \Sph(a,b)$, we mean  $t\in int\: \Sph(a,b)$.    
\\\\$\Rightarrow$) For $t \in \Sph(a,b)$, suppose that $\angle atb < \frac{\pi}{2}$ (proof by contradiction). We continue the line segment $\overline{at}$ to cross the boundary of the $\Sph(a,b)$. Let $t'$ be the crossing point (see Figure \ref{fig:point-circle-pi-in}). Since $\angle atb < \frac{\pi}{2}$, then, $\angle btt'$  is greater than $\frac{\pi}{2}$. Let $\angle btt'=\frac{\pi}{2}+\epsilon_1; \epsilon_{1}>0 $. From the Inscribed Angle Theorem, we know that $\angle at'b$ is a right angle. The angle $tbt'= \epsilon_{2}>0$ because $t\in \Sph(a,b)$. Summing up the angles in $\bigtriangleup tt'b$, as computed in (\ref{eq:point-circle-pi-in}), leads to a contradiction. So, this direction of proof is complete.

\begin{equation}
\angle tt'b + \angle t'bt + \angle btt'\geq \frac{\pi}{2}+\epsilon_{2}+ (\frac{\pi}{2}+\epsilon_{1}) = \pi +\epsilon_{1}+\epsilon_{2}> \pi
\label{eq:point-circle-pi-in} 
\end{equation}
\\\\$\Leftarrow$) If $\angle atb = \frac{\pi}{2}+\epsilon_{1}; \epsilon_{1}>0$, we prove that $t \in \Sph(a,b)$. Suppose that $t \notin \Sph(a,b)$ (proof by contradiction). Since $t \notin \Sph(a,b)$, at least one of the line segments $\overline{at}$ and $\overline{bt}$ crosses the boundary of $\Sph(a,b)$. Without loss of generality, assume that $\overline{at}$ is the one that crosses the boundary of $\Sph(a,b)$ at the point $t'$ (see Figure \ref{fig:point-circle-pi-out}). Considering the Inscribed Angle Theorem, we know that $\angle at'b=\frac{\pi}{2}$ and consequently, $\angle bt't=\frac{\pi}{2}$. The angle $\angle t'bt=\epsilon_{2}>0$ because $t \notin \Sph(a,b)$. If we sum up the angles in the triangle $\bigtriangleup tt'b$, the same contradiction as in (\ref{eq:point-circle-pi-in}) will be implied.
\end{proof}

\begin{figure}
\centering
\parbox{5cm}{
\includegraphics[width=5cm]{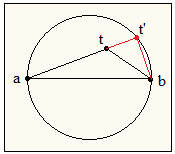}
\caption{$t$ inside $\Sph(a,b)$}
\label{fig:point-circle-pi-in}}
\qquad
\begin{minipage}{5cm}
\includegraphics[width=5cm]{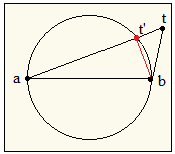}
\caption{$t$ outside $\Sph(a,b)$}
\label{fig:point-circle-pi-out}
\end{minipage}
\end{figure}

\paragraph{Algorithm:} Using Lemma \ref{lm:point-circle}, we present an algorithm to compute the spherical depth of a query point $q\in \mathbb{R}^2$ with respect to $S=\{x_1, x_2, ..., x_n\} \subseteq \mathbb{R}^2$.  This algorithm is summarized in the following steps.

\begin{itemize}
\item \textbf{Translating the points:}
Suppose that $T$ is a translation by $(-q)$. We apply $T$ to translate $q$ and all data points into their new coordinates. Obviously, $T(q)=0$.

 \begin{figure}[!ht]
  \centering
    \includegraphics[width=0.6\textwidth]{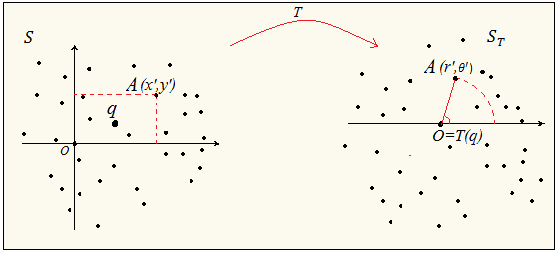}
  \caption{Transferred and sorted data set $S_T$ obtained from $S$}
  \label{fig:transferred-data-set}
\end{figure}

\item \textbf{Sorting the translated data points:} In this step we sort the translated data points based on their angles in their polar coordinates. After doing this step, we have $S_T$ which is a sorted array of the translated data points. See Figure \ref{fig:transferred-data-set}.

\item \textbf{Calculating the spherical depth:} Suppose that $x_i(r_i,\theta _i)$ is the $i^{th}$ element in $S_T$. For $x_i$, we define the arrays $O_i$ and $N_i$ as follows: 
\[ 
O_i=\left\{j\mid x_j\in S_T \: , \frac{\pi}{2} \leq |\theta_i -\theta_j|\leq \frac{3\pi}{2}\right\}
\]
\[ N_i= \{ 1,2,...,n\} \setminus O_i.\]
Thus the spherical depth of the origin of the coordinate system with respect to the data set $S_T$, which is equivalent to the spherical depth of $q$ with respect to $S$, can be computed by:
\[\SphD(q;S)=\SphD(T(q);S_T)=\SphD(0;S_T)= \frac{1}{2}\sum_{1\leq i\leq n}|O_i|,\]
where $|O_i|$ is the length of $O_i$. To present an formula for computing the $|O_i|$, we define $f_i$ and $l_i$ as follows:

\[f_i=
\begin{cases}
\min N_i -1 &\text{if $\frac{\pi}{2}< \theta_i \leq \frac{3\pi}{2}$}\\
    \min O_i & \text{otherwise}
\end{cases}
\]

\[l_i=
\begin{cases}
\max N_i +1 &\text{if $\frac{\pi}{2}< \theta_i \leq \frac{3\pi}{2}$}\\
    \max O_i & \text{otherwise.}
\end{cases}
\]

Figure \ref{fig:S-T-sorted-List} illustrates $O_i$, $N_i$,  $f_i$, and $l_i$ in two different cases. Considering the definitions of $f_i$ and $l_i$,
\[|O_i|=
\begin{cases}
f_i+(n-l_i+1) &\text{if $\frac{\pi}{2}< \theta_i \leq \frac{3\pi}{2}$}\\
    l_i-f_i+1 & \text{otherwise.}
\end{cases}
\]
\end{itemize}

\begin{figure}[!ht]
\centering
\parbox{6cm}{
\includegraphics[width=6.5cm]{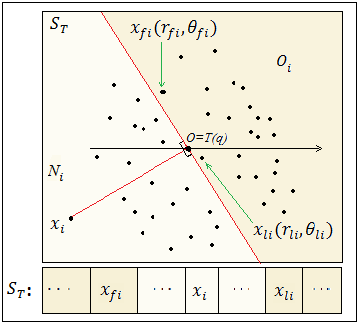}
\caption{Two representations of $S_T$ for $\theta\in (\frac{\pi}{2}, \frac{3\pi}{2}].$}
\label{fig:S-T-sorted-List}}
\qquad
\begin{minipage}{6cm}
\includegraphics[width=6.5cm]{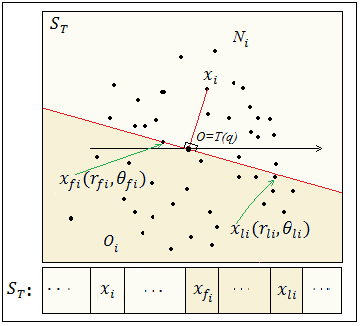}
\caption{Two representations of $S_T$ for $\theta\notin (\frac{\pi}{2}, \frac{3\pi}{2}].$}
\label{fig:S-T-sorted-List_2}
\end{minipage}
\end{figure}

\paragraph{Time complexity:} To analyse the time complexity of the algorithm, we need to compute the time complexities of the procedures in the algorithm (see the pseudocode of the algorithm in Algorithm \ref{Alg:sph- pseudocode}). The first procedure takes $O(n)$ to translate $q$ and all data points into the new coordinate system. The second procedure takes $O(n\log n)$ time. In this procedure, the loop iterates $n$ times, and the sorting algorithm takes $O(n\log n)$. Due to using the binary search algorithm for every $O_i$, the running time of the last procedure is also $O(n\log n)$. The rest of the algorithm contributes some constant time. In total, the running time of the algorithm is $O(n\log n)$.

\begin{algorithm}
\renewcommand{\algorithmicrequire}{\textbf{Input:}}
\renewcommand{\algorithmicensure}{\textbf{Output:}}
\newcommand{\Break}{\State \textbf{break}}
\caption{Computing the spherical depth of points in the plane}\label{Alg:sph- pseudocode}
\begin{algorithmic}[1]
\Require A data set $S$ and a query point $q$
\Ensure Spherical Depth of $q$ with respect to $S$ ($\SphD(q;S)$)
\item[]
\Procedure{Translating points}{}
\\input: $S$
\\output: Translated data set ($S_T$)
\For{each $x_i \in S$ }
\State $x_i \gets (x_i-q)$ 
\EndFor
\State \Return $S_T$
\EndProcedure
\item[]
\Procedure {Sorting translated data points around $T(q)$}{}
\\input: $S_T$ and $T(q)$
\\output: Sorted array $S_{TP}$
\For{each $x_i \in S_{T}$}
\State $x_i \gets$ Polar coordinate of $x_i$ \; \textit{//} $x_i(\theta_i,r_i)$ is obtained here.
\EndFor
\State Using an $O(n\log n)$ sorting algorithm, sort $x_i$ based on $\theta_i$
\State \Return $S_{TP}$
\EndProcedure
\item[]  
\Procedure {Depth calculation}{}
\\input: $S_{TP}$
\\output: Depth value of $\SphD(q;S)$
\For{each $x_{i}\in S_{TP}$}
\State $O_i \gets \left\{j\mid x_j\in S_{TP} \: , \frac{\pi}{2} \leq |\theta_i -\theta_j|\leq \frac{3\pi}{2}\right\}$
\State $N_i \gets \{1,2, \dots ,n \} \setminus O_i$
\State Using two \emph{binary search calls}, find the elements $f_i$ and $l_i$ in $S_{TP}$.
 
\State $f_i = 
\begin{cases}
\min N_i -1 &\text{if $\frac{\pi}{2}< \theta_i \leq \frac{3\pi}{2}$}\\
    \min O_i & \text{otherwise}
\end{cases}
$ 
\item[] 
\State $l_i = 
\begin{cases}
\max N_i +1 &\text{if $\frac{\pi}{2}< \theta_i \leq \frac{3\pi}{2}$}\\
    \max O_i & \text{otherwise.}
\end{cases}
$
 
\item[] 
\State Compute $|O_i|=
\begin{cases}
f_i+(n-l_i+1) ; &\text{if $\frac{\pi}{2}< \theta_i \leq \frac{3\pi}{2}$}\\
    l_i-f_i+1; & \text{otherwise.}
\end{cases}
$
\EndFor
\State $\SphD(q;S)\gets \frac{1}{2} \sum_{1\leq i \leq n}|O_i|$ 
\State \Return $\SphD(q;S)$
\EndProcedure 
\item[]
\State \textbf{End}; 
\end{algorithmic}
\end{algorithm}

\section{Lower Bound for Computing the Spherical Depth of a Point in the Plane}
We reduce the problem of Element Uniqueness to the problem of computing the spherical depth. It is known that the question of Element Uniqueness has a lower bound of $\Omega (n\log n)$ in the algebraic decision tree model of computation \cite{ben1983lower}.

\begin{theorem}
Computing the spherical depth of a query point in the plane takes $\Omega (n\log n)$ time.
\label{tm:lowebound}
\end{theorem} 

\begin{proof}
We show that finding the spherical depth allows us to answer the question of Element Uniqueness. Suppose that $A=\{a_1, a_2, ..., a_n\}$, for $n\geq 2$ is a given set of real numbers.\footnote{\textbf{Note:} We suppose all of the $a_i$s to be only positive(negative), otherwise we partition $A$ into two sets $A_1$ (consisting of only positive numbers) and $A_2$ (consisting of only negative numbers), and prove the theorem for $A_1$ and $A_2$, separately. Note that $A=A_1\cup A_2$ and $A_1 \cap A_2 = \emptyset$.} 
For every $a_i \in A$ we construct four points $x_i$, $x_{n+i}$, $x_{2n+i}$, and $x_{3n+i}$ in the polar coordinate system as follows:
\[ x_i=\left(r_i,\theta_i\right), \:
x_{n+i}=\left(r_i,\theta_i+\frac{\pi}{2}\right), \:
x_{2n+i}=\left(r_i,\theta_i+\pi\right), \: \text{and} \:
x_{3n+i}=\left(r_i,\theta_i+\frac{3\pi}{2}\right),\]
where $r_i= \sqrt{1+{a_i^2}}$ and $\theta_i=\tan^{-1}(1/a_i)$. Thus we have a set $S$ of $4n$ points $x_i$, $x_{n+i}$, $x_{2n+i}$, and $x_{3n+i}$, for $1\leq i \leq n$. See Figure \ref{fig:lower-bound-S}.
\\\\We select the query point $q=(0,0)$, and define $O_j$ as follows:
\begin{equation}
\label{eq:O-for-lowebound}
O_j=\left\{x_k\in S\mid \angle x_jqx_k \geq \frac{\pi}{2} \right\}, \: 1\leq j\leq 4n.
\end{equation}
We compute $\SphD(q;S)$ in order to answer the Element Uniqueness problem. Suppose that $x_j\in S$, for $1\leq j \leq 4n$, is a unique element. In this case, $|O_j|=2n+1$ because, from (\ref{eq:O-for-lowebound}), it can be figured out that the expanded $O_j$ is as follows:

\[
O_j=
\begin{cases}
\{x_{n+1},...,x_{n+j}, x_{2n+1},..., x_{3n}, x_{3n+j},...,x_{4n}\} ; &\text{if $j\in \{1,...,n\}$}\\
\{x_{2n+1},...,x_{n+j}, x_{3n+1},..., x_{4n}, x_{j-n},...,x_{n}\} ; &\text{if $j\in \{n+1,...,2n\}$}\\
\{x_{3n+1},...,x_{j+n}, x_{1},..., x_{n}, x_{j-n},...,x_{2n}\} ; &\text{if $j\in \{2n+1,...,3n\}$}\\
\{x_{1},...,x_{j-3n}, x_{n+1},..., x_{2n}, x_{j-n},...,x_{3n}\} ; &\text{otherwise.}
\end{cases}
\]

\begin{figure}[!ht]
  \centering
    \includegraphics[width=0.6\textwidth]{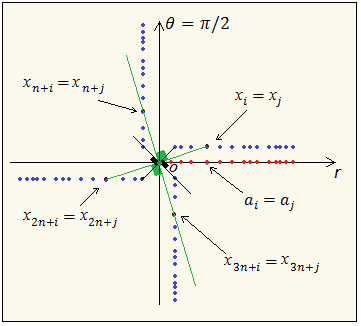}
  \caption{A representation of $A$ (determined by red points) $S$ (determined by blue points) and duplications in these sets}
  \label{fig:lower-bound-S}
\end{figure}

Referring to Lemma \ref{lm:point-circle} and Algorithm \ref{Alg:sph- pseudocode},
\[
\SphD(q;S)= \frac{1}{2}\sum_{1\leq j\leq 4n}|O_j|=\frac{1}{2}\sum_{1\leq j\leq 4n}(2n+1)=\frac{1}{2}(4n)(2n+1)=4n^2+2n\]
Now suppose that there exist some $i\neq j$ with $x_i= x_j$ in $S$. In this case, from (\ref{eq:O-for-lowebound}), it can be seen that 
\begin{gather*}
 |O_i|=|O_j|=|O_{(n+i)\bmod 4n}|=|O_{(n+j)\bmod 4n}|=|O_{(2n+i)\bmod 4n}|=\\ |O_{(2n+j)\bmod 4n}|=|O_{(3n+i)\bmod 4n}|=|O_{(3n+j)\bmod 4n}|=2n+2.  
\end{gather*}
See Figure \ref{fig:lower-bound-S}. Lemma \ref{lm:point-circle} and Algorithm \ref{Alg:sph- pseudocode} imply that
\begin{gather*}
\SphD(q;S)= \frac{1}{2}\sum_{1\leq j\leq 4n}|O_j|=\frac{1}{2}((4n-8)(2n+1) + 8(2n+2))= 4n^2+2n+4.
\end{gather*}
Thus if some element is duplicated in $S$, $\SphD(q;S)$ is strictly higher than if $S$ has no repetition in its elements. In other words, the question of Element Uniqueness can be answered by finding the spherical depth. Therefore the elements of $A$ are unique if and only if the spherical depth of $(0,0)$ with respect to $S$ is $4n^2+2n$. This implies that the computation of spherical depth require $\Omega (n\log n)$ time. It is necessary to mention that the only computations in the reduction are the construction of $S$ that take $O(n)$ time.   
\end{proof}

\section{Relationships Between Spherical Depth and Simplicial Depth}

\paragraph{Definition:}
For a point $q \in \mathbb{R}^2$ and a data set $S$ consisting of $n$ points in $\mathbb{R}^2$, we define $\Bin(q;S)$ to be the set of all closed sphere areas, out of $n \choose 2$ possible sphere areas, that contain $q$. We also define $\Sin(q;S)$ to be the set of all closed simplices, out of $n \choose 3$ possible closed simplices defined by $S$, that contain $q$.

\begin{lemma}
Suppose that $q$ is a point in a given convex hull $H$ obtained from a data set $S$ in $\mathbb{R}^2$. $q$ is covered by the union of sphere areas defined by $S$.
\label{lm:CH-q-GC}
\end{lemma}

\begin{proof}
It can be seen that there is at least one triangle, defined by the vertices of $H$, that contains $q$. We prove that the union of the sphere areas defined by such triangle contains $q$. See Figure \ref{fig:CH-q-GC} and Figure \ref{Triangleabc}. We prove this statement by contradiction. Suppose that $q$ is covered by none of $\Sph(a,b)$, $\Sph(a,c)$, and $\Sph(b,c)$. Therefore, Lemma \ref{lm:point-circle} implies that none of the angles $\angle aqb$, $\angle aqc$, and $\angle bqc$ is greater than or equal to $\frac{\pi}{2}$ which is a contradiction because at least one of these angles should be at least $\frac{2\pi}{3}$ in order to get $2\pi$ as their sum.
\end{proof}

\begin{lemma} Suppose that $S=\{a,b,c\}$ is a set of points in $\mathbb{R}^2$. For every $q\in \mathbb{R}^2$, if $|\Sin(q;S)|=1$, then $|\Bin(q;S)|\geq 2$.
\label{lm:triangle-ball}
\end{lemma}
Another form of Lemma \ref{lm:triangle-ball} is that if $q \in \bigtriangleup abc$, then $q$ falls inside at least two sphere areas out of three sphere areas $\Sph(a,b)$, $\Sph(c,b)$, and $\Sph(a,c)$. The equivalency between these two forms of the lemma is clear. We prove just the first one.

\begin{figure}[!ht]
\centering
\parbox{6cm}{
\includegraphics[width=6.5cm]{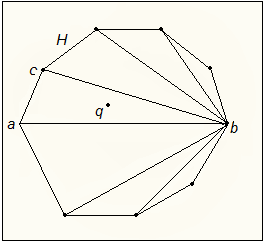}
\caption{Convex hull $H$ contains point $q$.}
\label{fig:CH-q-GC}}
\qquad
\begin{minipage}{6cm}
\includegraphics[width=6.5cm]{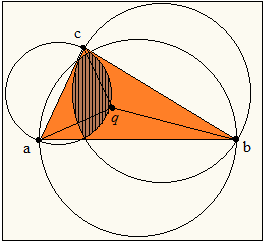}
\caption{Triangle $abc$ contains point $q$.}
\label{Triangleabc}
\end{minipage}
\end{figure}

\begin{proof}
We prove the lemma by contradiction. By Lemma \ref{lm:CH-q-GC}, $\Bin(q;S)\geq 1$. Suppose that $|\Bin(q;S)|=1$. If $q$ is located on the vertices of $\bigtriangleup abc$, it clear that $|\Bin(q;S)|\geq 2$ thus, we suppose that $q$ is not located on the vertices of $\bigtriangleup abc$. Without loss of generality, we suppose that $q$ falls inside the $\Sph(a,b)$. For the rest of the proof, we focus on the relationships among the angles $\angle aqb$, $\angle cqa $, and $\angle cqb$ (see Figure \ref{Triangleabc}). Since $q$ is inside $\bigtriangleup abc$, $\angle aqb \leq \pi$. Consequently, at least one of $\angle cqa$ and $\angle cqb$ is greater than or equal to $\frac{\pi}{2}$. So, Lemma \ref{lm:point-circle} implies that $q$ will fall inside at least one of the $\Sph(a,c)$ and $\Sph(b,c)$.  Hence, $|\Bin(q;S)|=1$ contradicts $|\Sin(q;S)|=1$. This means that the case $|\Bin(q;S)|\geq 2$. As an illustration, in Figure \ref{Triangleabc}, for the points inside the hatched area $|\Bin(q;S)|=3$.
\end{proof}

\begin{lemma} For $S=\{x_1,...,x_n\} \subset \mathbb{R}^2$, $\frac{|\Bin(q;S)|}{|\Sin(q;S)|}\geq \frac{2}{n-2}$.
\label{lm:bin-sin}
\end{lemma}

\begin{proof} We suppose that $\Sph(x_i,x_j)\in \Bin(q;S)$ (see Figure \ref{fig:bin-sin}). There exist at most $(n-2)$ triangles in $\Sin(q;S)$ such that $x_ix_j$ is an edge of them. Let us consider $ \bigtriangleup x_ix_jx_k$ from these triangles. Referring to Lemma \ref{lm:triangle-ball}, we know that $q$ falls inside at least one of $\Sph(x_i,x_k)$ and $\Sph(x_j,x_k)$. It means that there exist at most $(n-2)$ triangles in $\Sin(q;S)$ such that $x_ix_k$ (respectively $x_jx_k$) is an edge of them. As can be seen, the triangle $\bigtriangleup x_ix_jx_k$ is counted at least two times, one time for $\Sph(x_i,x_j)$ and one time for $\Sph(x_i,x_k)$ (or $\Sph(x_j,x_k)$ ). So, we can say that for every sphere area from $\Bin(q;S)$, such as $\Sph(x_i,x_j)$ there exist at most $\frac{(n-2)}{2}$ distinct triangles, triangles with only one common side, in $\Sin(q;S)$. Consequently, (\ref{eq:bin-sin}) can be obtained. 
 
\begin{equation}
\frac{(n-2)}{2}|\Bin(q;S)|\geq |\Sin(q;S)| \Rightarrow \frac{|\Bin(q;S)|}{|\Sin(q;S)|}\geq \frac{2}{(n-2)}
\label{eq:bin-sin}
\end{equation}
\end{proof}

\begin{figure}[h!]
  \centering
  \includegraphics[width=0.4\textwidth]{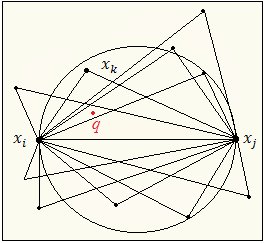}
  \caption{Sphere area $\Sph(x_i,x_j)$ contains point $q$}
  \label{fig:bin-sin}
\end{figure}

\begin{theorem}
For $q \in \mathbb{R}^2$ and a given data set $S$ which consists of $n$ points in $\mathbb{R}^2$, $\SphD(q;S)\geq \frac{2}{3} \SD(q;S)$.
\label{thrm:Sph-Simp}
\end{theorem}

\begin{proof}
From the definitions of spherical depth and simplicial depth, we can calculate the ratio of $\frac{\SphD(q;S)}{\SD(q;S)}$ as follows: 
\begin{equation}
\label{first-ratio1}
\frac{\SphD(q;S)}{\SD(q;S)}= \frac{\frac{|\Bin(q;S)|}{{n \choose 2}}}{\frac{|\Sin(q;S)|}{{n \choose 3}}}= \frac{|\Bin(q;S)|}{|\Sin(q;S)|} \times \frac{(n-2)}{3}
\end{equation}
From (\ref{first-ratio1}) and Lemma \ref{lm:bin-sin}, it can be seen that 
\[\frac{\SphD(q;S)}{\SD(q;S)} \geq \frac{2}{3} \Rightarrow \SphD(q;S) \geq \frac{2}{3} \SD(q;S).\]
\end{proof}

\section{Experiments}
To support Theorem \ref{thrm:Sph-Simp}, we compute the spherical depth and the simplicial depth of the points in three random sets $Q_1$, $Q_2$, and $Q_3$ with respect to data sets $S_1$, $S_2$, and $S_3$, respectively. The elements of $Q_i$ and $S_i$ are some randomly generated points (double precision floating point) within the square $A=\{(x,y)| x,y \in [-10,10]\}$. The results of our experiments are summarized in Table \ref{table:results-random-points}. Every cell in the table represents the corresponding depth of $q_i$ with respect to data set $S_i$, where $q_i \in Q_i$. The cardinalities of $Q_i$s and $S_i$s are as follows: $|Q_1|=100$, $|S_1|=750$, $|Q_2|=750$, $|S_2|=2500$, $|Q_3|=2500$, $|S_3|=10000$.
\\\\As can be seen in Table \ref{table:results-random-points}, the experimental results are consistent with Theorem \ref{thrm:Sph-Simp}. In fact, the experimental results suggest a bound that is stronger than the obtained bound in Theorem \ref{thrm:Sph-Simp}. This difference between the experimental bound and the theoretical bound is a motivation to do more research in this area.

\begin{table}[!ht]
\begin{center}
\begin{tabular}{|l||l|l||l|l||l|l|}
\hline
 &\multicolumn{2}{l|}{$(q_1;S_1)$}&\multicolumn{2}{l|}{$(q_2;S_2)$}&\multicolumn{2}{l|}{$(q_3;S_3)$}\\
\cline{2-7}
 &Min& Max&Min&Max&Min&Max\\
\hline\hline
$\SD$&0.00&0.25&0&0.25&0.00&0.24\\
\hline
$\SphD$&0.01&0.50&0.00&0.50&0.00&0.50\\
\hline
$\frac{\SphD}{\SD}$&2.00&$\infty$&2.00&$\infty$&2.02&$\infty$\\
\hline
\end{tabular}
\end{center}
\caption{Minimum and Maximum of \\simplicial depth,  spherical depth, and\\ the ratio of these data depths.}
\label{table:results-random-points}
\end{table}

\section{Conclusion}
In this paper, we developed an optimal $\Theta(n\log n)$ algorithm to compute the spherical depth of a bivariate query point with respect to a given data set in $\mathbb{R}^2$. To obtain a lower bound for the algorithm, the Element Uniqueness problem, which requires $\Omega (n\log n)$ time, is reduced to the computing of spherical depth. In addition to the time complexity, the main advantage of this algorithm is it simplicity for implementation. We also investigated some geometric properties which lead us to find a theoretical relationship (i.e. $\SphD\geq \frac{2}{3} SD$) between the spherical depth and the simplicial depth. Finally, some experimental results which suggest a stronger bound (i.e. $\SphD\geq 2SD$) are provided. More research on this topic is needed to figure out if the real bound is closer to the experimental bound or to the current theoretical bound.

\end{document}